\documentclass[a4paper,10pt]{article}

\usepackage{amsmath}
\usepackage{amsthm}
\usepackage{amssymb}
\usepackage[vlined,boxed]{algorithm2e}
\usepackage{graphicx}

\usepackage{color} 
\definecolor{tagada}{rgb}{1,0.15,0.15}

\usepackage{enumerate}
\usepackage{vmargin}

\setpapersize{A4}
\setmarginsrb{3cm}{2.5cm}{3cm}{2.5cm}{0.5cm}{0.5cm}{0.5cm}{1cm}

\theoremstyle{plain}
  \newtheorem{theorem}{Theorem}
  \newtheorem{lemma}[theorem]{Lemma}
  \newtheorem{corollary}[theorem]{Corollary}

\theoremstyle{definition}
  \newtheorem{definition}[theorem]{Definition}

\theoremstyle{remark}
  \newtheorem*{remark}{Remark}

\newcommand{\ids}{independent dominating set\xspace}
\newcommand{\MIDSpb}{\textsc{MIDS}\xspace}
\newcommand{\mids}{mids\xspace}
\newcommand{\MISpb}{\textsc{MIS}\xspace}
\newcommand{\MDSpb}{\textsc{MDS}\xspace}
\newcommand{\CSPpb}{\textsc{CSP}\xspace}

\newcommand{\branchm}{\mathbf{branch\_mark}}
\newcommand{\brancha}{\mathbf{branch\_all}}
\newcommand{\brancho}{\mathbf{branch\_one}}
\newcommand{\runtime}{1.3569}
\newcommand{\runtimeprec}{1.35684}

\title{A Branch-and-Reduce Algorithm for Finding a Minimum Independent
Dominating Set
\footnote{A preliminary version of this paper appeared in the
\textit{proceedings of WG 2006} \cite{WG2006}.}
}

\author{Serge~Gaspers\thanks{%
Center for Mathematical Modeling,
Universidad de Chile, 8370459 Santiago, Chile.
\texttt{sgaspers@dim.uchile.cl}}
\and Mathieu~Liedloff\thanks{%
Laboratoire d'Informatique Fondamentale d'Orl\'eans,
Universit{\'e} d'Orl{\'e}ans, 45067 Orl{\'e}ans Cedex 2, France.
\texttt{mathieu.liedloff@univ-orleans.fr}}
}

\date{}

\begin{document}

\SetAlgoSkip{}

\maketitle

\begin{abstract}
An \ids $\mathcal{D}$ of a graph $G=(V,E)$ is a subset of vertices such that
every vertex in $V \setminus \mathcal{D}$ has at least one neighbor in $\mathcal{D}$ and
$\mathcal{D}$ is an independent set, i.e. no two vertices of $\mathcal{D}$
are adjacent in $G$. Finding a minimum \ids in a graph is an NP-hard problem.
Whereas it is hard to cope with this problem using parameterized
and approximation algorithms, there is a simple exact $O(1.4423^n)$-time algorithm
solving the problem by enumerating all maximal independent sets.
In this paper we improve the latter result, providing the first non trivial algorithm computing a minimum \ids of a
graph in time $O(\runtime^n)$.
Furthermore, we give a lower bound of $\Omega(1.3247^n)$ on the worst-case running time of this algorithm,
showing that the running time analysis is almost tight.
\end{abstract}

\section{Introduction}

During the last years the interest in the design of
exact exponential time algorithms has grown significantly.
Several nice surveys have been written on this subject.
In Woeginger's first survey \cite{Woeginger}, he presents the
major techniques used to design exact exponential time algorithms.
We also refer the reader to the survey of Fomin et al. \cite{SurveyFGK}
discussing some more recent techniques for the design and the analysis of exponential time algorithms. 
In particular, they discuss Measure \& Conquer and lower bounds.

In a graph $G=(V,E)$, a subset of vertices $S\subseteq V$ is \emph{independent} if no two vertices
of $S$ share an edge, and $S$ is \emph{dominating} if every vertex from $V\setminus S$ has at least one neighbor
in $S$.
In the \textsc{Maximum Independent Set} problem (\MISpb), the input is a graph and the task is to find a largest
independent set in this graph. In the \textsc{Minimum Dominating Set} problem (\MDSpb), the input is
a graph and the task is to find a smallest dominating set in this graph.
A natural and well studied combination of these two problems asks for a subset of vertices
of minimum cardinality that is both independent and dominating. This problem is called
\textsc{Minimum Independent Dominating Set} (\MIDSpb).
It is also known as
\textsc{Minimum Maximal Independent Set}, since every \ids is a maximal independent set.
Whereas there has been a lot of work on \MISpb and \MDSpb
in the field of exact algorithms, the best known exact algorithm for
\MIDSpb\ -- prior to our work -- trivially enumerates all maximal independent sets.

\smallskip

\noindent {\bf Known results.}
The \MISpb problem 
was among the first problems shown to be NP-hard \cite{GareyJ79}.
It is known that a maximum independent set of a graph on $n$ vertices can be computed in $O(1.4423^n)$ time
by combining a result due to Moon and Moser, who showed in 1965 
that the number of maximal independent sets of a graph is upper
bounded by $3^{n/3}$ \cite{MoonMoser} (see also \cite{MillerM60}), and a result due to Johnson, Yannakakis and Papadimitriou,
providing in \cite{Johnson} a polynomial delay algorithm to generate all maximal
independent sets.
Moreover many exact algorithms for this problem have been published, starting
in 1977 by an $O(1.2600^n)$ algorithm by Tarjan and Trojanowski \cite{Tarjan}.
To date, the fastest known exponential space algorithms for \MISpb have been
designed by Robson. His algorithm from 1986 \cite{Robson} has running time $O(1.2108^n)$ and
his unpublished computer-generated algorithm from 2001 \cite{RobsonTR} has
running time $O(1.1889^n)$.
Among the currently leading polynomial space algorithms, there is a very simple algorithm with running
time $O(1.2210^n)$ by Fomin et al. \cite{soda2006,AcmFGK} from 2006,
an $O(1.2132^n)$ time algorithm by Kneis et al. \cite{KneisFSTTCS2009} from 2009,
and a very recent $O(1.2127^n)$ time algorithm by Bourgeois et al. \cite{BourgeoisSWAT2010}.

\smallskip

The \MDSpb problem
is also well known to be NP-hard \cite{GareyJ79}.
Until 2004, the only known exact exponential time algorithm to solve \MDSpb
asked for trivially enumerating the $2^n$ subsets of vertices. The year 2004 saw a particular
interest in providing some faster algorithms for solving
this problem. Indeed, three papers with exact algorithms for \MDSpb were published.
In \cite{WG2004} Fomin et al. present an $O(1.9379^n)$ time algorithm, in
\cite{Randerath} Randerath and Schiermeyer establish an $O(1.8899^n)$ time algorithm
and Grandoni \cite{Grandoni} obtains an $O(1.8026^n)$ time algorithm.

In 2005, Fomin et al. \cite{icalp2005,AcmFGK} use the Measure \& Conquer approach to obtain an algorithm with
running time $O(1.5263^n)$
and using polynomial space. By applying a memorization technique they show that this running
time can be reduced to $O(1.5137^n)$ when allowing exponential space usage.
Van Rooij and Bodlaender \cite{vanRooijSTACS2008} further improved the polynomial-space algorithm
to $O(1.5134^n)$ and the exponential-space algorithm to $O(1.5063^n)$.
By now, the fastest published algorithm is due to Van Rooij et al. In \cite{vanRooijESA2009}, they provide
a $O(1.5048^n)$ time needing exponential space to solve the more general counting version of \MDSpb, 
i.e. the problem of computing the number of distinct minimum dominating sets.

\medskip

It is known that a minimum \ids (a \mids, for short) can be found in polynomial time for
several graph classes like interval graphs \cite{Chang}, chordal graphs \cite{Farber},
cocomparability graphs \cite{Kratsch} and AT-free graphs \cite{Broersma}, whereas the problem
remains NP-complete for bipartite graphs \cite{Corneil} and comparability
graphs \cite{Corneil}.
Concerning approximation results, Halld\'orsson proved in \cite{Halldorsson}
that there is no constant $\epsilon>0$ such that \MIDSpb can be approximated
within a factor of $n^{1-\epsilon}$ in polynomial time, assuming $P\not =NP$.
The same inapproximation result even holds for circle graphs and bipartite graphs~\cite{Damian}.

The problem has also been considered in parameterized approximability.
Downey et al. \cite{Downey} have shown 
that it is $W[2]$-hard to approximate $k$-\textsc{Independent Dominating Set}
with a factor $g(k)$, for any computable function $g(k) \ge k$.
In other words, unless $W[2]=FPT$,
there is no algorithm with running time $f(k) \cdot n^{O(1)}$ (where $f(k)$ is any computable
function independent of $n$) which either asserts that there is no \ids
of size at most $k$ for a given graph $G$, or otherwise asserts
that there is one of size at most $g(k)$, for any computable function $g(k) \ge k$.

The first exponential time algorithm for \MIDSpb has been observed by Randerath and Schiermeyer \cite{Randerath}.
They use the result due to Moon and Moser \cite{MoonMoser} as explained
previously and an algorithm enumerating all the maximal independent sets
to obtain an $O(1.4423^n)$ time algorithm for \MIDSpb.
In 2006, an earlier conference version of this paper claimed an
$O(1.3575^n)$ time algorithm \cite{WG2006}. However, a flaw concerning 
the main reduction rule was discovered
by the authors and is repaired in the present paper.
Very recently, Bourgeois et al. \cite{BourgeoisSIROCCO2010} proposed a branch-and-reduce $O(1.3417^n)$ time algorithm,
reusing several of the ideas introduced in \cite{WG2006}.

\smallskip

\noindent {\bf Our results.}
In this paper we present an $O(\runtime^n)$ time algorithm for solving \MIDSpb
using the Measure \& Conquer approach to analyze its running time.
As the bottleneck
of the algorithm in \cite{Randerath} are the vertices of degree two, we develop several methods
to handle them more efficiently such as marking some vertices and a
reduction described in Subsection \ref{sec:cliques} to a constraint satisfaction problem.
Combined with some elaborated branching rules, this enables
us to lower bound shrewdly the progress made by the algorithm at each
branching step, and thus to obtain a polynomial-space algorithm with running time $O(\runtime^n)$.
Furthermore, we obtain a very close lower bound of $\Omega(1.3247^n)$ on the running
time of our algorithm, which is very rare for non trivial exponential time algorithms.

\smallskip

This paper is organized as follows. In Section 2, we introduce the necessary concepts and definitions.
Section 3 presents the algorithm for \MIDSpb. We prove
its correctness and an upper bound on its worst-case running time in Section 4.
In Section 5, we establish a lower bound on its worst-case running time, which is very close
to the upper bound
and we conclude with Section 6.

\section{Preliminaries}

Let $G=(V,E)$ be an undirected and simple graph. For a vertex $v\in V$ we denote 
by $N(v)$ the neighborhood of $v$ and by $N[v]=N(v)\cup \{v\}$ the closed neighborhood of $v$.
The degree $d(v)$ of $v$ is the cardinality of $N(v)$. For a given subset of vertices
$S\subseteq V$, $G[S]$ denotes the subgraph of $G$ induced by $S$, $N(S)$ denotes the set of
neighbors in $V\setminus S$ of vertices in $S$ and $N[S] = N(S) \cup S$. 
We also define $N_S(v)$ as $N(v) \cap S$, $N_S[v]$ as $N[v] \cap S$, and $d_S(v)$ (called the $S$-\emph{degree} of $v$) as the cardinality of
$N_S(v)$. In the same way, given two subsets of vertices $S\subseteq V$ and $X\subseteq V$, we define
$N_S(X)=N(X) \cap S$.

A \emph{clique} is a set $S \subseteq V$ of pairwise adjacent vertices.
A graph $G=(V,E)$ is \emph{bipartite} if $V$ admits a partition into two
independent sets.
A bipartite graph $G=(V,E)$ is \emph{complete bipartite} if every vertex
of one independent set is adjacent to every vertex of the other independent set.
A \emph{connected component} of a graph is a maximal subset of vertices inducing a
connected subgraph.

\medskip

In a branch-and-reduce algorithm, a solution for the current problem instance is computed by
recursing on smaller subinstances such that
an optimal solution, if one exists, is computed for at least one subinstance.
If the algorithm considers only one subinstance in a given case, we speak of a reduction rule,
otherwise of a branching rule.

Consider a vertex $u \in V$ of degree two with two non adjacent neighbors $v_1$ and $v_2$.
In such a case, a branch-and-reduce algorithm will typically branch into three subcases when considering $u$:
either $u$, or $v_1$, or $v_2$ are in the solution set. In the third branch, one can consider that
$v_1$ is not in the solution set as the second branch considers all solution sets containing $v_1$. In order
to memorize that $v_1$ is not in the solution set but still needs to be dominated, we mark $v_1$.

\begin{definition}
A \emph{marked graph} $G=(F,M,E)$ is a triple where 
$F\cup M$ denotes the set of vertices of $G$ and
$E$ denotes the set of edges of $G$.
The vertices in $F$ are called \emph{free vertices} and the ones in $M$
\emph{marked vertices}.
\end{definition}

\begin{definition}
\label{defidsmarked}
Given a marked graph $G=(F,M,E)$, an \ids $\mathcal{D}$
of $G$ is a subset of free vertices such that
$\mathcal{D}$ is an \ids of the graph $(F\cup M, E)$.
\end{definition}

\begin{remark}
It is possible that such an \ids does not exist in a marked graph,
for example if some marked vertex has no free neighbor.
\end{remark}

\noindent Finally, we introduce the notion of an \emph{induced marked subgraph}.

\begin{definition}
Given a marked graph $G=(F,M,E)$ and two subsets $S,T \subseteq (F\cup M)$, an
\emph{induced marked subgraph} $G[S,T]$ is the marked graph
$G'=(S, T, E')$
where $E' \subseteq E$ are the edges of $G$ with
both end points in $S \cup T$.
\end{definition}

\noindent Notions like neighborhood and degree in a marked graph $(F,M,E)$ are
the same as in the corresponding simple graph $(F \cup M,E)$.

\section{Computing a \mids on Marked Graphs}

In this section we present an algorithm
solving \MIDSpb on marked graphs, assuming that no marked vertex has $F$-degree larger than $4$.

{F}rom the previous definitions it follows that
a subset $\mathcal{D}\subseteq V$ is a \mids of a graph $G'=(V,E)$
if and only if $\mathcal{D}$ is a \mids of the
marked graph $G=(V,\emptyset,E)$. Hence the algorithm
of this section is able to solve the problem on simple graphs as well.
Also due to the definitions, edges incident to two marked vertices
are irrelevant; throughout this paper we assume that there are no such edges.

\medskip

Given a marked graph $G=(F,M,E)$, consider the graph $G[F]$ induced by
its free vertices. In the following subsection we consider the special case
when $G[F]$ is a disjoint union of cliques with some additional properties.

\subsection{$G[F]$ is a disjoint union of cliques}
\label{sec:cliques}

Assume in this subsection that the graph $G[F]$ is a disjoint union of cliques such that:
\begin{itemize}
\item each clique has size at most $4$, and
\item each marked vertex has at most $4$ free neighbors.
\end{itemize}

We will transform this instance $G=(F,M,E)$ of \MIDSpb into an instance
$(X,D,C)$ of the Constraint Satisfaction Problem (\CSPpb). Let us
briefly recall some definitions about \CSPpb.
Given a finite set $X= \{x_1, x_2, \dots, x_n\}$ of $n$ variables over domains $D(x_i)$,
$1\leq i \leq n$, and a set $C$ of $q$ constraints, \CSPpb asks for
an assignment of values to the variables, such that each variable is assigned a value from its domain,
satisfying all the constraints.
Formally, $(d,p)$-\CSPpb is defined as follows:

\begin{description}
\item[Input:] $(X,D,C)$ where $X=\{x_1, x_2, \dots , x_n\}$
is a finite set of variables over domains $D(x_i)$ of
cardinality at most $d$, $1 \leq i \leq n$,
and $C=\{c_1, c_2, ..., c_q\}$ is a set of constraints.
Each constraint $c_i \in C$ is a couple $\langle t_i,R_i\rangle$
where $t_i= \langle x_{i_1}, x_{i_2}, \dots, x_{i_j} \rangle$ is a $j$-tuple of variables, with $j \le p$,
and $R_i$ is a set of $j$-tuples of values over $D(x_{i_1}) \times D(x_{i_2}) \times \dots \times D(x_{i_j})$.

\item[Question:] Is there a function $f$ assigning to each variable $x_i \in X$, $1\leq i \leq n$,
a value of $D(x_i)$ such that for each constraint $c_i$, $1 \leq i \leq q$,
$\langle f(x_{i_1}), ..., f(x_{i_j}) \rangle \in R_i$ ?
\end{description}

Given a marked graph $G=(F,M,E)$ fulfilling the previous conditions,
we describe the construction of a $(4,4)$-\CSPpb instance.
We label the cliques $K_1, K_2, \dots, K_l$ of $G[F]$ respectively by $x_1, x_2, \dots, x_l$.
For each clique $K_i$, $1\leq i \leq l$, label its vertices from $v_i^1$ to $v_i^{|K_i|}$.
For each variable $x_i$, $1 \leq i \leq l$, we define its domain as $D(x_i)=\{1,2, \dots, |K_i|\}$.

Let $u_i \in M$ be a marked vertex and let $v_{i_1}^{k_1}, v_{i_2}^{k_2}, \dots, v_{i_j}^{k_j}$
be the free neighbors of $u_i$. Thus, $j\leq 4$.
Let $t_i=\langle x_{i_1}, x_{i_2}, \dots, x_{i_j}\rangle$ be the $j$-tuple of variables corresponding respectively
to the cliques containing $v_{i_1}^{k_1}, v_{i_2}^{k_2}, \dots, v_{i_j}^{k_j}$.
Let $R_i$ be the set of all $j$-tuples $\langle w_{i_1}, w_{i_2}, \dots, w_{i_j}\rangle$
over $D(x_{i_1}) \times D(x_{i_2}) \times \dots \times D(x_{i_j})$
such that for at least one $r$, $1\leq r \leq j$, the value of $w_{i_r}$ is $k_r$
and $\{u,v_{i_r}^{k_r}\}$ is an edge of the graph.

Finally, each marked vertex $u_i$ leads to a constraint $\langle t_i,R_i\rangle$
of the set $C$. Due to the conditions on the given marked graph,
the size of the domain of each variable is at most $4$
and the number of variables involved in each constraint is at most $4$.

%
%


We now use the following theorem of Angelsmark \cite{Angelsmark05}
showing that it is possible to restrict our attention to
$(2,4)$-\CSPpb.

\begin{theorem}[Theorem~11 of \cite{Angelsmark05}]
If there exists a deterministic $O(\alpha^n)$ time algorithm for
solving $(e,p)$-\CSPpb, then for all $d>e$, there exists a deterministic
$O((d/e + \epsilon)^n \alpha^n)$ time algorithm for solving
$(d,p)$-\CSPpb, for any $\epsilon >0$.
\end{theorem}

The constructive proof of this theorem
shows how to transform a $(d,p)$-\CSPpb instance on $n$ variables into a set of
$(e,p)$-\CSPpb instances on at most $n$ variables each, such that the $(d,p)$-\CSPpb instance has a solution
if and only if at least one of the $(e,p)$-\CSPpb instances has a solution. The number
of $(e,p)$-\CSPpb instances of this construction is bounded by $\Pi_{i>e} (i/e + \epsilon)^{n_i} \le (d/e+\epsilon)^{n}$,
where $n_i$ is the number of variables with domain size $i$ in the $(d,p)$-\CSPpb instance and $\epsilon >0$ can
be taken arbitrarily small.

We use this construction to transform our $(4,4)$-\CSPpb instance into a set of $\Pi_{i>2} (i/2 + \epsilon)^{N_i}$
$(2,4)$-\CSPpb instances, where $N_i$ is the number of cliques of size $i$ in $G[F]$.
Then, it is not hard to see that
there exists a \mids for $G$ if and only if at least one of the $(2,4)$-\CSPpb instances has
an assignment of the variables which satisfies all the constraints of this \CSPpb instance.
Given a satisfying assignment $f$ to such a \CSPpb instance,
the set $\bigcup_{i=1}^l \{ v_{i}^{f(x_i)} \}$
is a solution to \MIDSpb for $G$.
We obtain the following theorem.


\begin{theorem}
Let $N_2$, $N_3$ and $N_4$ be the number of variables (i.e. the number of cliques of $G[F]$)
with domain size (resp. of size) $2$, $3$ and $4$, respectively.
The corresponding \CSPpb instance can be solved in time
$O((4/2 + \epsilon)^{N_4} \cdot (3/2 + \epsilon)^{N_3} \cdot \alpha^{N_4+N_3+N_2})$
where $O(\alpha^n)$ is the running time needed to solve a $(2,4)$-\CSPpb instance on $n$ variables,
for any $\epsilon >0$.
\end{theorem}

The theorem can be combined with the following result of Moser and Scheder \cite{MoserS10}
providing an algorithm for solving $(2,4)$-\CSPpb.


\begin{theorem}[\cite{MoserS10}]\label{thm:MoserS10}
Any $(2,4)$-\CSPpb instance can be solved deterministically in time $O((1.5+\epsilon)^n)$, for any $\epsilon >0$.
\end{theorem}

\begin{corollary}\label{cor:csp}
Let $G=(F,M,E)$ be a marked graph such that
$G[F]$ is a disjoint union of cliques
of size at most $4$, and each marked vertex has $F$-degree at most 4.
Let $N_i$, $1\leq i \leq 3$, be the number of free vertices with $i$ free neighbors in $G$
(thus $G[F]$ has $N_i$ cliques of size $i+1$).
A \mids, if one exists, can be computed in time
$O( (1.5+\epsilon)^{N_1/2} \cdot (2.25 + \epsilon)^{N_2/3} \cdot (3+\epsilon)^{N_3/4})$
or it can be decided within the same running time that the marked graph has no \mids, for
any $\epsilon >0$.
\end{corollary}

We remark that the procedure of Corollary \ref{cor:csp} will not be a bottleneck in the final
running time analysis of our algorithm, even if we use the $1.6^n \cdot n^{O(1)}$ by Dantsin et al. \cite{Dantsin}
to solve $(2,4)$-\CSPpb instances instead of Theorem \ref{thm:MoserS10}.

\subsection{The Algorithm}
\label{sec:algo}

In this subsection, we give Algorithm {\bf ids} computing the size of a \mids
of a marked graph. Although the number of branching rules is quite large
it is fairly simple to check that the algorithm computes the size of
a \mids (if one exists). It is also not difficult to transform {\bf ids} into
an algorithm that actually outputs a \mids.
In the next section we prove the correctness and give a detailed analysis
of the running time of Algorithm {\bf ids}.

Once it has selected a vertex $u$, the algorithm makes recursive calls (that is, it branches)
on subinstances of the marked graph. There are different ways the algorithm branches and we
give the most common ones now. Let $v_1,\ldots ,v_{d_F(u)}$ denote the free neighbors of $u$, ordered by increasing $F$-degree.
The branching procedure $\brancha(G,u)$ explores all possibilities that $u$ or a free neighbor of $u$ is in the solution set. It returns
\begin{align*}
 1+\min_{v\in N_F[u]} \{ \mathbf{ids}(G[F\setminus N[v],M\setminus N(v)]) \}.
\end{align*}
The branching procedure $\branchm(G,u)$ additionally makes sure that the free neighbors of $u$ are considered by
increasing $F$-degree and when considering the possibility that $v_i$ is in the solution set, it marks all vertices $v_j$, $j<i$.
It returns
\begin{align*}
 1+\min \begin{cases}
	      \mathbf{ids}(G[F\setminus N[u],M\setminus N(u)]);\\
	      \min_{i=1..d_F(u)} \begin{cases}\mathbf{ids}(G[F\setminus (N[v_i] \cup \{v_1,\ldots,v_{i-1}\}),\\
	      \quad\quad\quad (M\cup \{v_1,\ldots,v_{i-1}\})\setminus N(v_i)]). \end{cases}
	    \end{cases}
\end{align*}
Finally, the branching procedure $\brancho(G,u)$ considers the two possibilities where $u$ is in the solution set or
where $u$ is not in the solution set. In the recursive call corresponding to the second possibility, $u$ is marked.
The procedure returns
\begin{align*}
 \min \begin{cases}
	      1+\mathbf{ids}(G[F\setminus N[u],M\setminus N(u)]);\\
	      \mathbf{ids}(G[F\setminus \{u\}, M\cup \{u\}]).
	    \end{cases}
\end{align*}


\begin{algorithm}[htbp]
{\small
\DontPrintSemicolon
\SetVlineSkip{0.44pt}
\Indm
\textbf{Algorithm ids($G$)\\}
\KwIn{A marked graph $G=(F,M,E)$ with $d_F(v)\le 4$ for each $v\in M$.}
\KwOut{The size of a \mids of $G$.}
\SetKwComment{LabRule}{}{}
\BlankLine
\Indp
   \If{$\exists u\in M \text{ s.t. } d_F(u)=0$} {
     \Return $\infty$\LabRule*[r]{$(1)$}
   }
   \ElseIf{$G[F]$ is a disjoint union of cliques} {
      \If{$\exists u\in F \text{ s.t. } d_F(u) \ge 5$} {
        \Return $\brancha(G,u)$\LabRule*[r]{$(2)$}
      }
      \ElseIf{$\exists u\in F \text{ s.t. } d_F(u) = 4$} {
        \Return $\brancho(G,u)$\LabRule*[r]{$(3)$}
      }
      \Else {
        \Return the solution determined by the algorithm of Corollary~\ref{cor:csp}\LabRule*[r]{$(4)$}
      }
   }
   \ElseIf{$\exists u \in M \text{ s.t. } d_F(u)=1$} {
      let $v$ be the free neighbor of $u$\;
      \Return $1+\mathbf{ids}(G[F \setminus N[v],M \setminus N(v)])$\LabRule*[r]{$(5)$}
   }
   \ElseIf{$\exists \text{ a connected component } B \text{ of } G[F] \text{ s.t. } |B|>2 \land G[B] \text{ is complete bipartite }$} {
      let $B$ be partitioned into two independent sets $X$ and $Y$\;
      \Return $\min \begin{cases}
	     |X|+\mathbf{ids}(G[F \setminus N[X],M \setminus N(X)]);\\
		 |Y|+\mathbf{ids}(G[F \setminus N[Y],M \setminus N(Y)])\end{cases}$\LabRule*[r]{$(6)$}
   }
   \ElseIf{$\exists C \subseteq F \text{ s.t. } |C| = 3 \land C \text{ is a clique } \land \exists! v \in C \text{ s.t. } d_F(v) \geq 3$} {
      \Return $\min \{
	      1+\mathbf{ids}(G[F\setminus N[v],M\setminus N(v)]);
	      \mathbf{ids}(G[F\setminus \{v\}, M])
	    \}$\LabRule*[r]{$(7)$}
   }
   \Else {
      choose $u \in F$ such that\;
      \Indp
        (a) $u$ is not contained in a connected component in $G[F]$ that is a clique,\;
        (b) according to (a), $u$ has minimum $F$-degree, and\;
        (c) according to (a) and (b), $u$ has a neighbor in $F$ of maximum $F$-degree.\;
      \Indm
      \If{$d_F(u) = 1$} {
         \Return $\brancha(G,u)$\LabRule*[r]{$(8)$}
      }
      \ElseIf{$d_F(u) = 2$}{
         \If{$u$ has a neighbor of $F$-degree at most $4$}{
             \Return $\branchm(G,u)$\LabRule*[r]{$(9)$}
         }
         \Else{
             \Return $\brancha(G,u)$\LabRule*[r]{$(10)$}
         }
      }
      \ElseIf{$d_F(u)=3$} {
         \If{all free neighbors of $u$ have $F$-degree 3}{
             Let $v\in N_F[u]$ such that $G[N_F(v)]$ has at most $1$ edge\;
             \Return $\brancho(G,v)$\LabRule*[r]{$(11)$}
         }
         \ElseIf{$u$ has a neighbor $v$ of $F$-degree 4} {
             \Return $\brancho(G,v)$\LabRule*[r]{$(12)$}
         }
         \ElseIf{$u$ has a neighbor $v$ of $F$-degree 5} {
            \Return $\min \begin{cases}
	           1+\mathbf{ids}(G[F \setminus N[u],M \setminus N(u)]);\\
               1+\mathbf{ids}(G[F \setminus N[v],M \setminus N(v)]);\\
			   \mathbf{ids}(G[F \setminus \{u,v\}, M \cup \{u,v\}])
			\end{cases}$\LabRule*[r]{$(13)$}
         }
         \ElseIf{$u$ has two free neighbors of $F$-degree 3} {
            \If{$N_F(u)$ is a clique} {
            	Let $v_3 \in N_F(u)$ with maximum $F$-degree\;
                \Return $\min \{
	                  1+\mathbf{ids}(G[F\setminus N[v_3],M\setminus N(v_3)]);
                      \mathbf{ids}(G[F\setminus \{v_3\}, M])
                \}$\LabRule*[r]{$(14)$}
            }
            \Else {
                \Return $\branchm(G,u)$\LabRule*[r]{$(15)$}
            }
	     }
         \Else {
            \Return $\brancha(G,u)$\LabRule*[r]{$(16)$}
         }
      }
      \ElseIf{$d_F(u)=4$} {
         \Return $\brancho(G,u)$\LabRule*[r]{$(17)$}
      }
      \Else(\tcp*[h]{$d_F(u) \ge 5$}){
         \Return $\brancha(G,u)$\LabRule*[r]{$(18)$}
      }
   }
}
\end{algorithm}

The branching procedure $\brancha$ is favored over $\branchm$ if $\branchm$ would create marked vertices of degree at least $5$.
Thus, starting with a graph where all the
marked vertices have $F$-degree at most $4$, Algorithm {\bf ids} will keep this invariant.
This property allows us to use the procedure described in the previous subsection whenever the
graph induced by its free vertices is a collection of cliques of size at most $4$.
The correctness and running time analysis of {\bf ids} are described in the next section.

\section{Correctness and Analysis of the Algorithm}
\label{sec:correctanaly}

In our analysis, we assign so-called weights to free vertices.
{F}ree vertices having only marked neighbors can be handled without
branching. Hence, it is an advantage when the $F$-degree of a vertex decreases.
The weights of the free vertices will therefore depend on their
$F$-degree.

Let $n_i$ denote the number of free vertices having $F$-degree $i$.
For the running time analysis we consider the following measure of the size of $G$:
\begin{equation*}
 k=k(G)=\sum_{i \geq 0} w_i n_i \leq n
\end{equation*}
with the weights $w_i \in [0,1]$.
In order to simplify the running time analysis, we make the following assumptions:

\begin{itemize}
\item $w_0=0$,
\item $w_i=1$ for $i \geq 3$,
\item $w_1 \leq w_2$, and
\item $\Delta w_1 \geq \Delta w_2 \geq \Delta w_3$ where $\Delta w_i = w_i-w_{i-1}, i \in \{1,2,3\}$.
\end{itemize}

\begin{theorem}
\label{upperbound}
Algorithm {\bf ids} solves \MIDSpb in time $O(\runtime^n)$.
\end{theorem}
\begin{proof}
An instance $I$ is \emph{atomic} if Algorithm \textbf{ids} does not make a recursive call on input $I$.
Let $P[k]$ denote the maximum number of atomic subinstances recursively processed to compute a
solution for an instance of size $k$. As the time spent in each call of {\bf ids}, excluding the time
spent by the corresponding recursive calls, is polynomial, except for Case (4), it is sufficient to show
that for a valid choice of the weights, $P[k]=O(\runtimeprec^k)$, and that the time spent in Case (4) does not exceed $P[k]$. Each recursive call made
by the algorithm is on an instance with at least one edge fewer, which means that the
running time of {\bf ids} can be upper bounded by a polynomial factor of $P[k]$. Moreover, as no reduction or
branching rule increases $k$, $P[k]$ can be bounded
by analyzing recurrences
based on the measure of the created subinstances in those cases where the algorithm makes
at least $2$ recursive calls.
We will analyze these cases one by one.

{\bf Case (1)}
A marked vertex that has no free neighbor cannot be dominated. Thus, such an instance has no \ids.

{\bf Case (2)}
In this case, $G[F]$ is a disjoint union of cliques and $u$ is a vertex from a clique
of size $\ell \ge 6$ in $G[F]$. The branching $\brancha(G,u)$ creates $\ell$ subinstances
whose measure is bounded by $k-\ell w_3$. The corresponding recurrence relation is
$
P[k] \leq \ell P[k- \ell w_3].
$
For $\ell \ge 6$, the tightest of these recurrences is when $\ell = 6$:
\begin{equation}
P[k] \leq 6 P[k- 6 w_3].
\end{equation}

{\bf Case (3)}
In this case, $G[F]$ is a disjoint union of cliques and $u$ is a vertex from a clique
of size $5$ in $G[F]$. The branching $\brancho(G,u)$ creates $2$ subinstances
whose measure is bounded by $k-5 w_3$ and $k-w_3$, respectively. Note that the marked vertex which is created in the second branch has $F$-degree $4$. The corresponding recurrence is
\begin{equation}
P[k] \leq P[k- 5 w_3] + P[k- w_3].
\end{equation}

{\bf Case (4)} The graph induced by the free vertices is a disjoint union of cliques of size no more
than $4$. Corollary~\ref{cor:csp} is applied on the remaining marked graph and we note that
the number $n_i$ of vertices of $F$-degree $i$, $1\leq i \leq 3$, in this graph is no more than
$n_1\leq \mu /w_1 \leq n /w_1$,
$n_2\leq \mu /w_2 \leq n /w_2$ and
$n_3\leq \mu /w_3 \leq n /w_3$
with $n_1+n_2+n_3 \leq n$.

{\bf Case (5)}
A marked vertex $u$ with exactly one free neighbor $v$ must be dominated by $v$. Thus, $v$ is added to the \mids and all
its neighbors are deleted.

{\bf Case (6)}
If there is a subset $B$ of free vertices such that $G[B]$ induces a complete bipartite
graph and no vertex of $B$ is adjacent to a free vertex outside $B$, then the algorithm branches
into two subcases. Let $X$ and $Y$ be the two maximal independent sets of $G[B]$. Then
a \mids contains either $X$ or $Y$. In both cases we delete $B$
and the marked neighbors of either $X$ or $Y$.
The smallest possible subset $B$ satisfying the conditions of this case
is a $P_3$, that is a path on three vertices,
as $|B|>2$. Note that all smaller complete bipartite graphs are cliques and will be handled by Case (4).
Since we only count the number of free vertices, we obtain the following recurrence:
\begin{equation}
P[k] \leq 2P[k-2w_1-w_2].
\end{equation}
It is clear that any complete bipartite component with more than three
vertices would lead to a better recurrence.

{\bf Case (7)}
If there is a subset $C$ of three free vertices which form a clique and exactly one vertex $v \in C$
has free neighbors outside $C$, the algorithm either includes $v$ in the solution set or it
excludes this vertex. In the first branch, all the neighbors of $v$ are deleted
(including $C$). In the second branch, note that $v$ is not marked. Indeed, $v$'s $F$-degree
might be too high to be marked, and $v$'s neighborhood contains a clique component in $G[F]$ of which one vertex is in every
\ids of the resulting marked graph, making the marking of $v$ superfluous.
We distinguish two cases based on the number of free neighbors of some free vertex $u \in N(v) \setminus C$.
\begin{enumerate}
\item Vertex $u$ has one free neighbor. In the first branch, all of $N[v]$ are deleted, and in the second branch, $v$ is removed,
$u$'s $F$-degree decreases to $0$, and the $F$-degree of both vertices in $C \setminus \{v\}$ decreases to $1$.
This gives the recurrence:
\begin{equation}
P[k] \leq P[k-w_1-2w_2-w_3]+P[k+w_1-2w_2-w_3].
\end{equation}
\item Vertex $u$ has $F$-degree at least $2$. Then we obtain the recurrence:
\begin{equation}
P[k] \leq P[k-3w_2-w_3]+P[k+2w_1-2w_2-w_3].
\end{equation}
\end{enumerate}

{\bf Case (8)}
If there is a free vertex $u$ such that $d_F(u)=1$, a \mids either includes $u$
or its free neighbor $v_1$. Vertex $v_1$ cannot have $F$-degree one
because this would contradict the first choice criterion (a) of $u$. For the analysis, we
consider two cases:
\begin{enumerate}
\item $d_F(v_1)=2$. Let $x_1$ denote the other free neighbor of $v_1$. Note that 
$d_F(x_1) \not = 1$ as this would have been handled by Case (6).
We consider again two subcases:
\begin{enumerate}
\item $d_F(x_1)=2$. When $u$ is chosen in the \ids, $u$ and $v_1$ are deleted and the
degree of $x_1$ decreases to one. When $v_1$ is chosen in the \ids, $u,v_1$ and $x_1$
are deleted from the marked graph. So, we obtain the following recurrence for
this case:
\begin{equation}
P[k] \leq P[k-2w_2] + P[k-w_1-2w_2].
\end{equation}

\item $d_F(x_1) \geq 3$. Vertices $u$ and $v_1$ are deleted in the first branch,
and $u$, $v_1$ and $x_1$ are deleted in the
second branch. The recurrence for this subcase is:
\begin{equation}
P[k] \leq P[k-w_1-w_2] + P[k-w_1-w_2-w_3].\label{rec:tight1}
\end{equation}
\end{enumerate}

\item $d_F(v_1) \geq 3$. At least one free neighbor of $v_1$ has $F$-degree at
least 2, otherwise Case (6) would apply. Therefore the recurrence
for this subcase is:
\begin{equation}
P[k] \leq P[k-w_1-w_3] + P[k-2w_1-w_2-w_3].
\end{equation}
\end{enumerate}

{\bf Case (9)}
If there is a free vertex $u$ such that $d_F(u)=2$ and $u$ has a neighbor
of $F$-degree at most $4$ (as the neighbors $v_1,v_2$ of $u$ are ordered by increasing
$F$-degree, $v_1$ has $F$-degree at most $4$), the algorithm uses $\branchm(G,u)$ to branch into three subcases. 
Either $u$ belongs to the \mids, or $v_1$ is
taken in the \mids, or $v_1$ is marked and $v_2$ is taken in the \mids.
We distinguish three cases:
\begin{enumerate}
\item $d_F(v_1)=d_F(v_2)=2$. In this case, due to the choice of the vertex $u$
by the algorithm, all free vertices of this connected component $T$ in
$G[F]$ have $F$-degree 2. $T$ cannot be a $C_4$ (a cycle on 4 vertices)
as this is a complete bipartite graph and would have
been handled by Case (6).
In the branches where $u$ or $v_1$ belong to the \mids, the three free vertices
in $N[u]$ or $N[v_1]$ are deleted and two of their neighbors ($T$ is a cycle on
at least $5$ vertices) have their $F$-degree reduced from $2$ to $1$.
In the branch where $v_1$ is marked and $v_2$ is added to the \mids, 
$N[v_2]$ is deleted and by Case (5), the other neighbor $x_1$ of $v_1$ is
added to the \mids, resulting in the deletion of $N[x_1]$ as well. In total,
at least $5$ free vertices of $F$-degree $2$ are deleted in the third branch.
Thus, we have the recurrence
\begin{equation}
P[k] \leq 2 P[k+2w_1-5w_2] + P[k-5w_2]
\end{equation}
for this case.

\item $d_F(v_1)=2, d_F(v_2) \ge 3$. The vertices $v_1$ and $v_2$ are not adjacent,
otherwise Case (7) would apply. In the last branch, $v_1$ is marked and $v_2$ is added
to the solution. If $v_1$ and $v_2$ have a common neighbor besides $u$, then the last
branch is atomic because Case (1) applies as no vertex can dominate $v_1$. Otherwise, the
reduction rule of Case (5) applies in the last branch and the other neighbor $x_1 \not = u$
is added to the solution as well. Thus, we have the recurrence
\begin{equation}
P[k] \leq P[k-2w_2-w_3] + P[k-3w_2] + P[k-5w_2-w_3].
\end{equation}

\item $3\le d_F(v_1) \le 4$. We distinguish between two cases depending on whether there 
is an edge between $v_1$ and $v_2$.
\begin{enumerate}
\item $v_1$ and $v_2$ are not adjacent. Branching on $u$, $v_1$ and $v_2$
leads to the following recurrence:
\begin{equation}
P[k] \leq P[k-w_2-2w_3]+P[k-3w_2-w_3]+P[k-3w_2-2w_3].
\end{equation}

\item $v_1$ and $v_2$ are adjacent. We distinguish two subcases depending on whether there is a degree-$2$ vertex in $N^2(u)$.
\begin{enumerate}
\item There is a degree-$2$ vertex in $N^2(u)$. Then,
\begin{equation}
P[k] \leq P[k+w_1-2w_2-2w_3]+2P[k-2w_2-2w_3].
\end{equation}
\item No vertex in $N^2(u)$ has degree $2$. Then,
\begin{equation}
P[k] \leq P[k-w_2-2w_3]+2P[k-w_2-3w_3].\label{rec:tight2}
\end{equation}
\end{enumerate}
\end{enumerate}
\end{enumerate}

{\bf Case (10)}
If there is a free vertex $u$ such that $d_F(u)=2$ and none of the above cases
apply, then $v_1$ and $v_2$ have degree at least $5$ and the algorithm branches into the 
three subinstances of $\brancha(G,u)$: either $u, v_1,$ or $v_2$ belongs to the \mids, leading to the recurrence
\begin{equation}
P[k] \leq P[k-w_2-2w_3]+2P[k-5w_2-w_3].
\end{equation}

{\bf Case (11)}
If all neighbors of $u$ have degree $3$, then the connected component in $G[F]$ containing 
$u$ is $3$-regular due to the selection criteria of $u$. As (by criterion (a)) this 
component is not a clique, $N_F^2(u)$ is not empty. Thus, there exists some $v\in N_F[u]$ 
such that $G[N_F(v)]$ has at most one edge. This means that there are at least $4$ edges 
with one endpoint in $N_F(v)$ and the other endpoint in $N_F^2(v)$. If $|N_F^2(v)|=2$, the recurrence corresponding to the 
branching $\brancho(G,v)$ is
\begin{equation}
P[k] \leq P[k+2w_1-6w_3]+P[k+3w_2-4w_3],
\end{equation}
if $|N_F^2(v)|=4$ it is
\begin{equation}
P[k] \leq P[k+4w_2-8w_3]+P[k+3w_2-4w_3],\label{rec:tight3}
\end{equation}
and if $|N_F^2(v)|=3$ it is a mixture of the above two recurrences and is majorized by one
or the other.

{\bf Case (12)}
If $u$ has a neighbor $v$ of $F$-degree $4$, then the algorithm uses the branching procedure
$\brancho(G,v)$. If $v$ is taken in the \mids, $5$ vertices of degree at least $3$ are
removed from the instance. If $v$ is marked, the $F$-degree of $u$ decreases from $3$ to $2$. The corresponding recurrence is
\begin{equation}
P[k] \leq P[k-5w_3]+P[k+w_2-2w_3].
\end{equation}

{\bf Case (13)}
If $u$ has a neighbor $v$ of $F$-degree $5$, then the algorithm either takes $u$ in the 
\mids, or $v$, or it marks both $u$ and $v$ (note that $v$ will have $F$-degree 4). The recurrence corresponding to this case is
\begin{equation}
P[k] \leq P[k-4w_3]+P[k-6w_3]+P[k-2w_3].\label{rec:tight4}
\end{equation}

{\bf Case (14)}
In this case, $N_F[u]$ is a clique and $v_3$ is the only vertex from this clique that has
free neighbors outside $N_F[u]$. The algorithm either takes $v_3$ in the \mids or deletes 
it. Note that $N_F(v_3)$ includes a clique and that any \mids of $G[F\setminus 
\{v_3\},M]$ contains one vertex from this clique, which makes the marking of $v_3$ superfluous. 
\begin{equation}
P[k] \leq P[k-7w_3]+P[k+3w_2-4w_3].
\end{equation}

{\bf Case (15)}
We distinguish two cases based on the neighborhood of $v_3$.
\begin{enumerate}
\item $v_3$ is adjacent to $v_1$ and $v_2$. Then, $v_1$ is not adjacent to $v_2$, otherwise
Case (14) would apply. In the second branch, $v_2$'s $F$-degree drops to $1$ and in the
third branch, $v_1$'s neighbor in $N_F^2(u)$ is also selected by Case (5). This gives the
recurrence
\begin{equation}
P[k] \leq P[k-4w_3]+P[k+w_1-5w_3]+P[k-5w_3]+P[k-7w_3].
\end{equation}
\item $v_3$ is not adjacent to $v_1$ or to $v_2$. In the last branch, $7$ vertices are deleted and one vertex is marked, giving
\begin{equation}
P[k] \leq 3P[k-4w_3]+P[k-8w_3].
\end{equation}
\end{enumerate}

{\bf Case (16)}
In this case, $u$ has at least two neighbors of degree at least $6$. The recurrence corresponding to the branching $\brancha(G,u)$ is
\begin{equation}
P[k] \leq 2P[k-4w_3]+2P[k-7w_3].
\end{equation}

{\bf Case (17)}
If $u$ has degree $4$, the algorithm branches along $\brancho(G,u)$, giving the recurrence
\begin{equation}
P[k] \leq P[k-5w_3]+P[k-w_3].
\end{equation}

{\bf Case (18)}
If $u$ has degree $\ell \ge 5$, the algorithm branches along $\brancha(G,u)$. The corresponding recurrence is $P[k] \leq (\ell+1) P[k-(\ell+1)w_3]$, the tightest of which is
obtained for $\ell=5$:
\begin{equation}
P[k] \leq 6P[k-6w_3].
\end{equation}

Finally the values of weights are computed with a convex optimization program \cite{GaspersS09} (see also \cite{Gaspers08})
to minimize the bound on the running time.
Using the values $w_1=0.8482$ and $w_2=0.9685$ for the weights, one can
easily verify that $P[k]=O(\runtimeprec^k)$. In particular by this choice of the weights, the running-time required
by Corollary~\ref{cor:csp} to solve the \CSPpb instance whenever Case (2) is applied is no more than $O(1.3220^k)$
(it would be bounded by $O(1.3517^k)$ if we used the algorithm of Dantsin et al. \cite{Dantsin} for solving (2,4)-\CSPpb).
Thus, Algorithm \textbf{ids} solves \MIDSpb in time $O(\runtime^n)$.
\end{proof}

The tight recurrences of the latter proof (i.e. the worst case
recurrences) are (\ref{rec:tight1}), (\ref{rec:tight2}), (\ref{rec:tight3}), and (\ref{rec:tight4}).

\section{A Lower Bound on the Running Time of the Algorithm}
In order to analyze the progress of the algorithm during the computation of a
\mids, we used a non standard measure. In this way we have been able to
determine an upper bound on the size of the subinstances recursively processed by the
algorithm, and consequently we obtained an upper bound on the
worst case running time of Algorithm \textbf{ids}. However the use of another measure
or a different method of analysis could perhaps provide a
``better upper bound'' without changing the algorithm but only improving the
analysis.

How far is the given upper bound of Theorem~\ref{upperbound} from the best upper
bound we can hope to obtain?
In this section, we establish a lower bound on the worst case running time of
our algorithm. This lower bound gives a really good estimation on the precision of
the analysis. For example, in \cite{icalp2005} (see also \cite{AcmFGK}) Fomin et al. obtain a
$O(1.5263^n)$ time algorithm for solving the dominating set problem and
they exhibit a construction of a family of graphs giving a lower bound of $\Omega(1.2599^n)$
for its running time. They say that the upper bound of many exponential time
algorithms is likely to be overestimated only due to the choice of the measure
for the analysis of the running time, and they note the gap between their upper
and lower bound for their algorithm.
However, for our algorithm we have the following result:

\begin{theorem}
\label{lowerbound}
Algorithm {\bf ids} solves \MIDSpb in time $\Omega(1.3247^n)$.
\end{theorem}

To prove Theorem~\ref{lowerbound} on the lower bound of the worst-case running time of
algorithm {\bf ids},
consider the graph $G_l=(V_l,E_l)$ (see Fig. \ref{lbgraph}) defined by:
\begin{itemize}
\item $V_l = \{ u_i, v_i : 1 \leq i \leq l\}$,
\item $E_l =\{u_1,v_1\} \cup \big\{ \{ u_i,v_i\},\{u_i,u_{i-1}\}, \{v_i,v_{i-1}\}, \{u_i,v_{i-1}\} : 2 \leq i \leq l \big\}$.
\end{itemize}

We denote by $G'_l=(V,\emptyset,E)$ the marked graph corresponding
to the graph $G_l=(V,E)$.

\begin{figure}
\centering
      \includegraphics[scale=0.55]{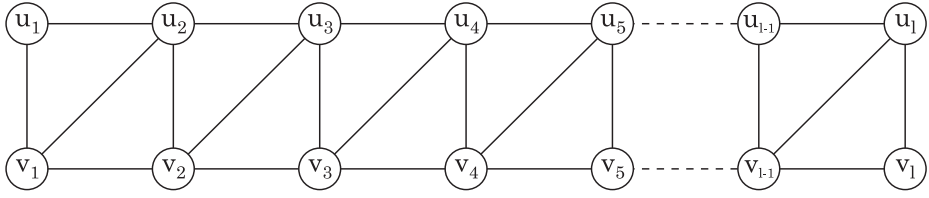}
      \caption{\label{lbgraph}graph $G_l$}
\end{figure}

For a marked graph $G=(F,M,E)$ we define
$\delta_F = \min_{u \in F} \{d_F(u)\}$
and
$MinDeg = \{u \in F \text{ s.t. } d_F(u) = \delta_F \}$ as the set of free vertices
with smallest $F$-degree.

We denote the highest $F$-degree of the free neighbors of the vertices in $MinDeg$ by
$\Delta_{\delta_F} = \max \big\{ d_F(v) : v \in N_F(MinDeg) \big\}$.

Let $CandidateCase9= \{ u \in MinDeg: \exists v \in N_F(u) \text{ s.t. } d_F(v) = \Delta_{\delta_F} \}$
be the set of candidate
vertices that {\bf ids} can choose in Case (9).
W.l.o.g. suppose that when $|CandidateCase9|\geq 2$ and {\bf ids} would
apply Case (9), it chooses the vertex
with smallest index (e.g. if $CandidateCase9 = \{u_1,v_l\}$, the algorithm would
choose $u_1$).

\begin{lemma}\label{LBLemma1}
Let $G'_l$ be the input of Algorithm {\bf ids}.
Suppose that {\bf ids} only applies Case (9) in each recursive call
(with respect to the previous rule for choosing a vertex).
Then, in each call of {\bf ids} where the remaining input graph has more
than four vertices, one of the following two properties is fulfilled:

\begin{enumerate}[(1)]
\item $CandidateCase9 = \{ u_k, v_l \}$ for a certain $k$, $1\leq k \leq l-2$, and
\begin{enumerate}[(i)]
\item the set of vertices $\bigcup_{1\leq i<k} \{u_i,v_i\}$ has been deleted from
the input graph, and
\item all vertices in $\bigcup_{k\leq i\leq l} \{u_i,v_i\}$ remain free in
the input graph.
\end{enumerate}

\item $CandidateCase9 = \{ v_k, v_l \}$ for a certain $k$, $1\leq k \leq l-2$, and
\begin{enumerate}[(i)]
\item the set of vertices $\{u_k \} \cup \bigcup_{1\leq i<k} \{u_i,v_i\}$ has
been deleted from the input graph, and
\item all vertices in $\{v_k \} \cup \bigcup_{k< i\leq l} \{u_i,v_i\}$ remain
free in the input graph.
\end{enumerate}
\end{enumerate}
\end{lemma}

\begin{proof}
We prove this result by induction. It is not hard to see that $Candidate\-Case9=\{u_1,v_l\}$
for $G'_l$ and that Property (1) is verified.

Suppose now that Property (1) is fulfilled. Then there exists an integer $k$, $1\leq k \leq l-1$,
such that $CandidateCase9 = \{ u_k, v_l \}$.
Since {\bf ids} applies Case (9) respecting the rule for choosing the vertex in $CandidateCase9$,
the algorithm chooses vertex $u_k$.
Then we branch on three subinstances:
\begin{enumerate}[(b1)]
\item Take $u_k$ in the \mids and remove $N[u_k]$.
Thus, the remaining free vertices are $\{v_{k+1} \} \cup \bigcup_{k+1< i\leq l} \{u_i,v_i\}$
whereas all other vertices are removed. Moreover for this remaining
subinstance, we obtain $Candidate\-Case9=\{ v_{k+1}, v_l \}$. So, Property
(2) is verified.
(Note also that $|N[u_k] \cap \bigcup_{k\leq i\leq l} \{u_i,v_i\}| = 3$.)

\item Take $v_k$ in the \mids and remove $N[v_k]$:
$\bigcup_{k+2\leq i\leq l} \{u_i,v_i\}$ is the set of the remaining free vertices
and all other vertices are removed. For the remaining
subinstance we obtain $Candidate\-Case9=\{ u_{k+2}, v_l \}$
and Property (1) is verified.
(Note also that $|N[v_k] \cap \bigcup_{k\leq i\leq l} \{u_i,v_i\}| = 4$.)

\item Take $u_{k+1}$ in the \mids and remove $N[u_{k+1}]$. Thus,
the remaining free vertices are $\{v_{k+2} \} \cup \bigcup_{k+2< i\leq l} \{u_i,v_i\}$
and all other vertices are removed. For this remaining
subinstance we obtain $CandidateCase9=\{ v_{k+2}, v_l \}$
and Property (2) is verified.
(Note also that $|N[u_{k+1}] \cap \bigcup_{k\leq i\leq l} \{u_i,v_i\}| = 5$.)
\end{enumerate}

If we suppose now that Property (2) is fulfilled, branching on
a vertex $v_k$ gives us the same kind of subproblems.
\end{proof}

Now, we prove that, on input $G_l$,
Algorithm {\bf ids} applies Case (9) as long as
the remaining graph has ``enough'' vertices.

\begin{lemma}
Given the graph $G'_l$ as input,
as long as the remaining graph has more than four vertices,
Algorithm {\bf ids} applies Case (9) in each recursive call.
\end{lemma}

\begin{proof}
We prove this result also by induction.
First, when the input of the algorithm is the graph $G'_l$,
it is clear that none of Cases (1) to (8) can be applied.
So, Case (9) is applied since $CandidateCase9\not = \emptyset$ according
to Lemma~\ref{LBLemma1}.

Consider now a graph obtained from $G'_l$ by repeatedly
branching using Case (9).
By Lemma~\ref{LBLemma1}, the remaining graph
has no marked vertices (this excludes that
Cases (1) and (5) are applied). It has no clique component induced by the set of free vertices
since the graph is connected and there is no edge between
$u_{l-1}$ and $v_l$ (this excludes Cases (2)--(4)). The free vertices do not induce a bipartite graph
since $\{ v_{l-1}, u_l, v_l\}$ induces a $C_3$ (this excludes Case (6)). There is no clique $C$
such that only one vertex of $C$ has neighbors outside $C$: the largest induced clique in the remaining
graph has size 3 and each of these cliques has at least two vertices having some neighbors outside
the clique (this excludes Case (7)). Also, according to Lemma~\ref{LBLemma1}, the remaining graph
has no vertex of degree 1 (this excludes Case (8)) and $CandidateCase9\not = \emptyset$.
Consequently, the algorithm applies Case (9).
\end{proof}

\begin{figure}[htb]
\centering
      \includegraphics[scale=0.6]{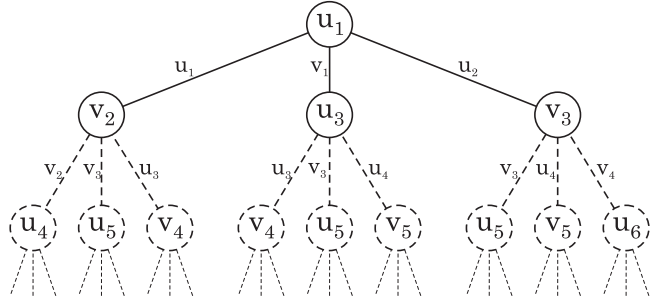}
      \caption{\label{lbsearchtree}a part of the search tree}
\end{figure}

Figure \ref{lbsearchtree} gives a part of the search tree illustrating the fact that
our algorithm recursively branches on three subinstances with respect to Case (9).

\begin{proof}[Proof of Theorem~\ref{lowerbound}]
Consider the graph $G'_l$ and the search tree which results from branching using
Case (9) until $k$ vertices, $1\leq k \leq 2l$, have been removed from the
given input graph $G'_l$ ($G'_l$ has $2l$ vertices).
Denote by $L[k]$ the number of leaves in this search tree.
It is not hard to see that this leads to the following recurrence
(see the notes in the proof of Lemma \ref{LBLemma1}):
$$L[k] = L[k-3] + L[k-4] + L[k-5]$$ and therefore $L[k] \geq 1.3247^{k}$.
Consequently, the maximum number of leaves that
a search tree for {\bf ids} can contain, given an input graph on $n$
vertices, is $\Omega(1.3247^n)$.
\end{proof}

\section{Conclusions and Open Questions}

In this paper we presented a non trivial algorithm solving
the \textsc{Minimum Independent Dominating Set} problem. Using a non standard measure on
the size of the considered graph, we proved that our algorithm achieves a running
time of $O(\runtime^n)$. Moreover we showed that $\Omega(1.3247^n)$ is a lower bound on
the running time of this algorithm by exhibiting a family of graphs for which
our algorithm has a high running time.

A natural question here is: is it is possible to obtain a better upper bound on
the running time of the presented algorithm by considering another measure or
using other techniques. Or is it possible that this upper bound is tight?



\end{document}